\documentclass[10pt,a4paper]{article}
\usepackage[utf8]{inputenc}
\usepackage[english]{babel}
\usepackage{csquotes}
\usepackage[figuresright]{rotating}
\usepackage[color=pink]{todonotes}
\usepackage{subcaption}
\usepackage{amsmath,amsfonts,amssymb,mathtools,amsthm}
\usepackage{makecell}
\usepackage{soul}
\usepackage{graphicx}
\usepackage{apacite}
\usepackage[sort&compress,semicolon,square]{natbib}
\usepackage{pgfplots}
\usepackage{caption}
\usepackage{breakcites}
\usepackage{appendix}
\usepackage{multirow}
\usepackage{array}
\usepackage{bm}
\usepackage{enumitem}
\usepackage{tikz}
\usepackage{pgfplots}
\usepackage{soul}

\usepackage[section]{placeins}

\usetikzlibrary{shapes,snakes,shapes.geometric,patterns, positioning}
\tikzset{
	dots size/.store in=\dotssize,
	dots size=1pt,
	dots spread/.store in=\dotsspread,
	dots spread=10pt
}
\tikzstyle{mybox} = [draw=black, fill=white, very thick, rectangle, rounded corners, inner sep=15pt, inner ysep=25pt] 
\tikzstyle{fancytitle} =[draw=black, fill=white, rounded corners, text=black, inner xsep=4pt, inner ysep=6pt, font=\sc] 

\makeatletter

\newtheorem{theorem}{Theorem}
\newtheorem{lemma}{Lemma}

\newtheorem{proposition}{Proposition}

\newtheorem{assumption}{Assumption}

\makeatother

\title{A core-selecting auction for portfolio's packages}

\author{\textsc{Lamprini Zarpala} \thanks{ I am genuinely grateful to the 33rd Stony Brook International Conference on
Game Theory audience and the 2022 Conference on Mechanism and Institution Design for their valuable comments. Additionally, I am genuinely grateful to
the Royal Economic Society for awarding me a conference grant as an early-stage researcher. Send correspondence to l.zarpala@uu.nl.} \\Utrecht University (U.S.E.) \and \textsc{Dimitris Voliotis}\footnote{dvoliotis@unipi.gr}\\University of Piraeus }

\date{ }

\providecommand{\keywords}[1]{\textbf{Keywords - }#1}
\providecommand{\JEL} [1] {\textbf{JEL Classification -} #1}

\begin{document}
	\maketitle
	\begin{abstract}
We introduce the \enquote{local-global} approach for a divisible portfolio and perform an equilibrium analysis for two variants of core-selecting auctions.  Our main novelty is extending the Nearest-VCG pricing rule in a dynamic two-round setup, mitigating bidders' free-riding incentives and reducing the sellers' costs. The two-round setup admits an information-revelation mechanism that may offset the \enquote{winner's curse}, and it is in accord with the existing iterative procedure of combinatorial auctions. With portfolio trading becoming an increasingly important part of investment strategies, our mechanism contributes to increasing interest in portfolio auction protocols.

\end{abstract}
\begin{flushleft}	
	\keywords Package auction $\cdot$ VCG payments $\cdot$ Portfolio Trading
	\\
	\JEL D44 $\cdot$ D47 $\cdot$ G11
\end{flushleft}

\section{Introduction}
Portfolio auctions are a type of auction where multiple assets are bundled and offered for sale as a single unit. These auctions are often employed to sell large and complex collections of assets, offering an efficient and low-risk approach compared to traditional auctions. Across finance, real estate, energy contracts, and art, portfolio auctions have become increasingly prevalent.

In the financial sector, portfolio auctions typically involve a basket of financial assets and have gained prominence due to the growth of ETFs and algorithmic trading. These auctions are designed as an automated execution protocol,\footnote{ICE Bonds, Tradeweb, and OneChronos are all examples of portfolio auctions for securities. These platforms allow asset managers and brokers to sell large and complex collections of securities transparently and efficiently.} that facilitates the seamless auction and trading of a portfolio of securities, such as stocks and bonds. The execution can be done either by trading the entire portfolio at once or by \enquote{slicing} it into optimal packages using algorithmic trading strategies. The latter may achieve a better execution price than the single execution, yet it is possible to increase the execution time \citep{Gsell2008Assessing,KONISHI2002197}.


Portfolio auctions are also known as basket or program trading portfolio investment strategies. These strategies have been in the investment landscape for a while and have accounted for 50\% to 60\% of the total daily trading volume at the NYSE\footnote{In periods with high volatility, it can reach up to 90\%. Since 2012, the NYSE no longer publishes weekly program trading reports.}, allowing broker-dealers to bid fast for a basket of trades executed as a single transaction. Before submitting bids, broker-dealers possess a clear understanding of each asset's weight in the basket, expressed either as a dollar value or percentage \citep{Giannikos2007a}.

The brokerage plays an important role in basket trading strategy using two alternative strategies:  principal and agency trading\footnote{In principal trading, the brokerage completes a customer's trade with their inventory and can profit from the bid-ask spread, while in agency trading, the brokerage locates a counterparty to a customer's trade willing to purchase or sell the security for the same price as the counterparty.}. We focus on \emph{principal} trade, in which a broker undertakes the risk price on behalf of the asset manager and executes the portfolio at an agreed price plus a commission fee. A key aspect of principal trading is the need for the broker to minimize any deficit caused by the price differential, i.e., the actual execution price versus the agreed price with the asset manager. The commission fee for the broker serves as compensation for undertaking the risk of executing the portfolio according to the agreed-upon terms and as a source of revenue for the services provided. 

The impact of this investment strategy for asset managers is vast, especially in periods of high volatility where cost savings and speedy risk transfer are imperative. Even though asset managers always have the option to execute those trades at their discretion, using a portfolio auction trading tool enables them to have access to multiple liquidity providers simultaneously. In addition, this strategy reduces any information leakage and guarantees execution efficiency - an efficient way to deal with large and complex transactions. 
From the broker's perspective, this strategy is an opportunity to access new order flows that might match their inventory or facilitate new business. The challenge is to bid low enough to win the auction and at the same time to cover the assumed price risk from the execution \citep{Padilla2012,Giannikos2012a}.

Portfolio auctions typically employ single-round first-price sealed-bid format where the entire portfolio is allocated to a single broker. One of the caveats of this design is that it can cause weak demand since some brokers might have no preference for possessing the whole portfolio but have a higher valuation for a part of it. Thus, the asset manager might suffer an increased cost for the portfolio's execution.  Another issue is that it does not provide bidders with valuable information to avoid aggressive bidding, thereby the \enquote{winner’s curse}. Indeed, brokers may overestimate the portfolio's value, potentially paying more than it is worth and causing a financial loss. Brokers may possess limited information or need more research capabilities to estimate the portfolio's value accurately. Hence, the winner's curse phenomenon is often manifested in portfolio auctions. 

Motivated by this context and other real-world auctions, we examine here a package auction for portfolios. We assume that the auctioneer divides the portfolio into packages for a discrete number of brokers with different valuations. To avoid auction's inefficiency and the \enquote{exposure problem},\footnote{In environments with complementarities, simultaneous bidding for multiple packages exposes bidders with strong complementarities to the risk of winning by bidding more for a particular package than it is worth on its own \citep{Rothkopf1998,Roth2003}} the packages are preset such that they minimize the portfolio's execution cost and market impact losses. We characterize those who compete for the packages as \enquote{local} brokers and those who bid for the whole portfolio as \enquote{global} brokers. Each broker knows the individual securities in the portfolio and takes long positions in all securities. Our two-round design with information updates in the interim creates incentives for brokers to reduce their bidding further and increase the asset manager's revenue.

 We propose a new pricing rule, the D-NVCG, that aligns the bidder's incentives with truthful bidding in both rounds and mitigates any free-riding opportunities that brokers might have if the auction occurred as a single round \citep{Baranov2010, Roth2003}. A local broker in the first round might get the desired package by bidding aggressively, with the expectation of receiving a higher fee in the second round if the coalition underbids the global broker. The rule considers the first-round bid as a reference point to adjust the final fees paid to the local brokers. We prove that this pricing rule is optimal for the coalition of brokers. 

The auction occurs in two rounds. In the first round, the asset manager simultaneously performs sealed-bid auctions (a) for \enquote{local} brokers who compete for each package separately and (b) for \enquote{global} brokers who compete for the aggregate portfolio. The first round qualifies those brokers with the lowest fee (first price) - one global broker and one local broker for each package - to participate in the second round.  The practice of qualifying one single global broker may puzzle some observers. It might seem to limit competition in the second round and harm the seller. However, the qualification stage promotes only those brokers with the lowest bid in the next round, i.e., the most competitive, and brokers' qualification bids set an upper bound for the next round.

In the interim, the asset manager reveals information for the qualified bids of the first round. This attribute in the design provides a solution to brokers' information asymmetries by revealing valuations and avoiding jump-bidding \citep{Perry2000}. This private information update is valuable for the price discovery of brokers, as they can understand the likely price for each package and the whole portfolio. Thus, their valuation efforts are more productive, and the auction becomes more efficient and less susceptible to manipulations \citep{Ausubel2017, Milgrom2021}

In the second round, the local brokers jointly compete in a sealed-bid auction for the whole portfolio against the global, who values packages as perfect complements.\footnote{The utility of the whole portfolio has a higher utility than the sum of the utilities for the individual packages \citep{Crampton2006}.} If the coalition of local brokers submits a bid lower than that of global, the coalition wins, and the fees are awarded to each local broker based on a pricing rule that ensures a \textit{core} outcome. 

In our setup, the \textit{core} attributes the set of payoff vectors corresponding to the portfolio's allocation, where no better outcome exists for the asset manager and the local brokers. Any payoff vector in the core is \textit{ bidder-optimal} if there is no other payoff vector that \textit{Pareto} improves upon it, i.e., packages are allocated to bidders with the highest allocation.
In other words, the coalition of local brokers is unblocked and feasible \citep{Day2008b}. 
Hence, our core-selecting auction provides a framework where the payoff of brokers is on the \emph{bidder-optimal-frontier} \citep{milgrom2004}, and may reduce the costs of the asset manager significantly due to the two-round setup  \citep{Perry2000}. In the existing literature, both the D-NVCG rule and the N-VCG rule lead to results in the bidder-optimal frontier. However, the D-NVCG includes a boundary related to the first-round bidding behavior that may further lower the costs for asset managers.

From a practical perspective, the D-NVCG can be applied in general package auctions such as spectrum auctions, energy contracts, or even real estate auctions.  Its implementation as a pricing rule in discrete iterative auctions effectively incentivizes truthful bidding during the final two rounds. In iterative auctions, bidders reveal partial and indirect information about their valuation. The D-NVCG addresses one of the biggest challenges of iterative auctions: strategic bidding and economic efficiency \citep{parkes2006iterative}. The activity rule \citep{milgrom2004} that D-NVCG comprises in the pricing induces brokers to bid low from the first stage of the auction. 

\subsection{Outline}
Our paper proceeds as follows. Section \ref{literature} reviews the related literature. Section \ref{Model} presents the model and describes the two-round mechanism. Section \ref{Analysis} derives the intuitive form of the optimality conditions and analyses the equilibrium. Finally, section \ref{concludes} discusses the results and concludes.


\section{Related Literature} \label{literature}

Most literature in package auctions (or combinatorial auctions) focuses on developing fast heuristics to solve the complex winner's determination problem \citep{Rothkopf1998, DeVries2003, Benedikt}, while economists focus on specific properties by using simplified theoretical models \citep{Day2008b,Milgrom2007, Ausubel2019}. 

The underpinnings of package auctions can be traced to the seminal paper of \cite{Vickrey1961b}, in which each bidder is asked to pay an amount equal to the externalities he exerts on the competing bidders. Vickrey shows that this payment rule motivates bidders to submit a \enquote{bid} according to the actual demand schedules, regardless of the bids made by others.  
It is easy to think that a Vickrey auction could generate an efficient outcome for package auctions due to its appealing property of incentive compatibility. Nevertheless, in practice, the Vickrey auction is rarely used because this mechanism can lead to low payoffs for the auctioneer, even if bids are high enough \citep{Milgrom2007,Ausubel2019}. Also, the Vickrey pricing is determined by a non-monotonic function of the broker's values in the sense that an increase in the number of brokers can reduce equilibrium revenues for the asset manager up to zero. Thus, brokers can use profitably \enquote{shill bidding} to increase competition in order to finally charge higher fees \citep{Ausubel2002,Milgrom2006Ausubel}.

The existing literature alleviates the aforementioned shortcomings by proposing alternative procedures.
\cite{Ausubel2002} developed a mechanism called the ascending proxy auction, while, \cite{Xia} reviewed several pricing schemes for incentive compatibility and ascertained that the revelation of the losing bids may reduce the value of the prices relative to winning bids.
\cite{Day2008b}  and \cite{Day2012} suggested a new cluster of payment rules for core-selecting auctions with respect to the  reported values. 

Recently, \cite{Ausubel2019} provided a theoretical justification for using core-selecting auctions. They propose an incomplete-information setting in which bidders' values are correlated and analyze the equilibrium under a \enquote{local-global} approach. They found that in environments with positive correlations, core-selecting auctions can be significantly closer to the true core than the VCG outcome\footnote{\cite{Goeree2016} show that a mechanism that guarantees a core outcome for true values does not always exist.}. To our knowledge, \cite{Krishna1996a} were the first to explore an independent private value setting\footnote{\cite{Rosenthal1996} extended the setting with common values.} for the simultaneous sale of multiple items in the \enquote{local-global} setting. 

The problem that arises in the \enquote{local-global} setting is when a coalition wins and the VCG outcome is not in the core.\footnote{ If the VCG outcome is in the core, it is the unique bidder-optimal allocation.}

Then, the closer the bidder-optimal frontier gets to VCG pricing, the fewer incentives for misreporting \citep{Day2008b, Ausubel2002}. \cite{Day2007a} and \cite{Day2012} find alternative payment rules that minimize the bidder's incentives for this strategic manipulation. This rule is called the \emph{Nearest-VCG (Quantratic Rule)}, a point in the bidder-optimal-frontier where the maximum deviation from VCG pricing is minimized. A critical assumption of the above research is that packages are equally weighted. Yet, equal sharing of the core outcome might be unfair for the small allocations, so another option is to share the core outcome with the weighted Nearest-VCG rule \citep{Ausubel2017}. The weighted Nearest-VCG rule is more relevant for portfolio auctions where each package has a different weight over the portfolio, and any pricing without encountering the weight would be unfair.

\cite{Erdil2017a} have proposed a new class of pricing rules for core-selecting package auctions focusing on the marginal incentives to deviate from \enquote{truthful bidding}. The idea is to select a point in the bidder-optimal frontier close to a reference point. Motivated by their suggestion, we construct a new payment rule, the Dynamic-Nearest-VCG, using an endogenous reference point suggested by the brokers' strategic bidding behavior in the first round. The rule finds a point in the bidder-optimal frontier by calculating the distance from the first-round bid. By including the first-round bid in the final payments, the mechanism becomes robust since it provides incentives to bidders for \enquote{truthfull} reporting in both rounds.

\cite{Benedikt} developed a computational search approach to design bidder-optimal-frontier rules, which are at the facet of the core polytope under a local-global setup. They find that some rules based on Sharpley value outperform the Nearest-VCG, yet their computation is complex. We investigate how to restore truthful incentives from the first round while they show that the best-performing rules are those that provide incentives to bidders with large values.

From a different perspective, it has been developed a vast literature on iterative combinatorial auctions\footnote{A multiple-round bidding process at which the auctioneer releases information regarding the provisional winners and the actual prices at the end of each round. Bidders obtain information regarding the bids of their rivals and can modify their bids in the following rounds \citep{Parkes2006}.}.
One of the merits of this approach is its ease of deployment. It allows bidders to learn about rivals' valuations, and it is the most popular combinatorial auction format used in practice.  For example, the FCC has used only multi-round formats for its auction design \citep{Porter2003}.  Moreover, auction designs that allow the creation of synergies through bidder-determined combinations\footnote{In this context, bidders have the authority to decide what is biddable and which combinations hold economic significance to them.} can yield economically impactful results, sidestepping potential computational complexities \citep{Park2005}.

The Combinatorial Clock Auction (CCA)  in the spectrum auctions is relevant to our model and bears one similarity\citep{ausubel_cramton_milgrom_2017,cramton2009auctioning,bichler2013core}. The similarity lies in the two-stage bidding process with information feedback in the interim. Yet, the whole design is different in both stages. In the qualifying round of our model, the auctioneer specifies the packages, while in the CCA design, bidders specify the desired packages at reserve prices. In the final round of our model the coalition of local brokers bids against the broker who is interested in the whole portfolio, while in the CCA bidders place multiple supplementary bids, which improve their clock-round bids and express values for other packages. \citep{Ausubel2017}.

\cite{levin20} identify a weakness in the equilibrium outcome of CCA since bidders in the first round (clock phase) can exaggerate their bids and gain a surplus in the supplementary round through types' revelation which might lead to a strategically increase in the rival's costs \citep{janssen2019clock}. Thus, in the clock phase, bidders' incentives might be distorted from the trade-off between efficiency and information revelation. 

CCA even though it mitigates the communication complexity of the commonly used bidding\footnote{In XOR bidding language bidders specify values for all the possible combinations of items \citep{Bichler2023}} language, it constrains the expressiveness of the bids and may cause bidders to win only a subset of the desired items at prices exceeding their valuation, known as \enquote{exposure problem} \citep{Bichler2023,Roth2003}. The auction format that we adopt with the preset packages from the beginning allows bidders to state their preferences concretely and simplifies any assumption for computationally complexity\footnote{The number of possible allocations grows exponentially, and with many objects for sale, there is no guarantee the optimal allocation is found in a feasible time} that would arise if bidders specified the packages \citep{Roth2003,Rothkopf1998, GOEREE2010146,Scheffel2012}.

An interesting extension of our model would be to allow brokers to bid simultaneously for multiple packages in the first round and then participate in the second round by consolidating the packages' weights from the qualifying round.  While this approach presents computational challenges and necessitates efficient management of the portfolio's execution time, it could yield a dual outcome. The asset manager might receive competitive bids from the first round, yet brokers might hedge on the packages that do not value, which likely would have a market impact on spot prices for the asset manager and the brokers \citep{Padilla2012}.

\section{Model} \label{Model}

An asset manager sells $m > 2$ securities packaged in a divisible
portfolio $\Theta\in \mathbb{R}_+^m$. 
A set of risk-neutral
brokers, $\mathbb{N}=\{1,\dots,n\}$ are competing\footnote{We assume no entry costs for participating brokers.} for a portfolio $\Theta$, with $n\geq2$.
Two types of brokers participate in the auction: a set of \textit{local} brokers denoted by the generic element $\ell\in \mathbb{L}$, and a set of \textit{global} brokers $\mathit{g}\in \mathbb{G}$, where $\mathbb{L},\mathbb{G}\subset \mathbb{N}$  are disjoint and $\mathbb{L} \cup \mathbb{G} = \mathbb{N}$.

We assume that $\Theta$ is divided in a finite set of $q$ packages $\theta_j$, with $j=\{1,\dots,q$\}, such as $\theta_j\in \mathbb{R}_+^m$ and $\Theta=\sum\limits_{j=1}^{q}\theta_j$. 
The vector $p^*\in \mathbb{R}^m_+$ includes the agreed
exercise price for $m$ securities, and the vector $p\in \mathbb{R}_+^m $ the anticipated price of $m$ securities when delivered.
We denote by $\omega_j=\dfrac{p^*\cdot \theta_j}{p^*\cdot \Theta}$, the weight of $\theta_j$'s value over $\Theta$ with $\omega_j\in(0,1]$ and $\omega_j\leq \frac{1}{n-1}$. Even for the extreme cases, the definition of $\omega$ holds, and our results remain unaltered.

Each broker $i\in \mathbb{N}$ observes a
private signal $s_i\in \mathbb{S}$ about the value of the $\Theta$ or $\theta_j$ and a public signal $\mathit{z}\in \mathbb{Z}$ for the aggregate characteristics of the portfolio. 
We denote by $s_{\ell}$ the private signal of local broker $\ell$, where $s_{-\ell}$ aggregates the private signal of all other local brokers and  $s_g$ the signal of global broker accordingly.
Information $(\mathbf{s},\mathit{z})$, where $\mathbf{s}=(s_{\ell},s_{-\ell},s_g)$, $\forall \ell \in \mathbb{L} $ and $g\in \mathbb{G}$,
is distributed according to a continuous i.i.d. function $F_{\ell}(\cdot)$, with the density function $f_{\ell}(\cdot)>0$  for \textit{local} brokers, and $F_\mathit{g}(\cdot)$, with $f_\mathit{g}(\cdot)>0$ for \textit{global} brokers respectively.  

All brokers form expectations for the percentage change of
the securities' prices given by the vector,
\[\mathbb{E}\Big[\frac{\Delta p}{p^*}|s_i,z\Big]=\bigg(\dfrac{p_k^*-\mathbb{E}[p_k|s_i,z]}{p_k^*}\bigg)_{k\leq m},\]
where the vector $p^*\in \mathbb{R}^m_+$ includes the agreed
exercise price for $m$ securities, and the vector $p\in \mathbb{R}_+^m $ the anticipated price of $m$ securities when delivered. Both random vectors, $p^*\in \mathbb{R}^m_+$ and $p\in \mathbb{R}_+^m $, are conditional on the signal received and the available public information. 
Evidently, when $\mathbb{E}[\frac{\Delta p}{p^*}|s_i,z]>0$ brokers anticipate to incur a loss. 


\textbf{Example:} Lets assume three securities, $m_1$,$m_2$,$m_3$, with $m_1=6$, $m_2=7$ and $m_3=14$ items respectively in the aggregate portfolio, such as $\Theta=(6,7,14)$.\\ The agreed execution price of each security is: 

\small\[  p^*=\begin{pmatrix}
		
		p^*_1 \\ p^*_2\\ p^*_3 
	\end{pmatrix}=\begin{pmatrix}
		
		2.54 \\ 4.89 \\ 3.10
	\end{pmatrix}\].  \\ 
 Also, we assume that the portfolio is divided into three packages\footnote{Suppose that package $\theta$ consists of only one security. In this case, $\theta$ can be represented as a sparse vector with only one non-zero element. For example, if the security is denoted as $m_1$, the sparse vector representation would be $\theta = (m_1, 0, 0, \ldots, 0)$.} :
	
	\small\[ \Theta = \left\{\begin{array}{lr}
		\theta_1=(2,1,3  )\\
		\theta_2=(1,4,5)\\  \theta_{3}=(3,2,6)
	\end{array}\right\} = \sum_{i=1}^{3}\theta_i \] 

Then the weight of each package is with $\omega_1=\dfrac{p^*\cdot \theta_1}{p^*\cdot \Theta}\approx 0.2075$, $\omega_2=\dfrac{p^*\cdot \theta_2}{p^*\cdot \Theta}\approx 0.4049$ and $\omega_3=\dfrac{p^*\cdot \theta_3}{p^*\cdot \Theta}\approx 0.3876$.

Now lets assume that the anticipated prices of $m_1$,$m_2$,$m_3$ securities  when delivered are:

       \small \[p=\begin{pmatrix}
		
		\mathbb{E}[p_1] \\\mathbb{E}[p_2]\\ \mathbb{E}[p_3] 
	\end{pmatrix}=\begin{pmatrix}
		
		2.80 \\ 4.35 \\ 3.50 
	\end{pmatrix}\]
 Brokers form expectations for the percentage change of the securities' prices:
	    	\small \[\mathbb{E}\Big[\frac{\Delta p}{p^*}|s_i,z\Big]=
		\bigg(\dfrac{p_k^*-\mathbb{E}[p_k|s_i,z]}{p_k^*}\bigg)_{k\leq m}=\begin{pmatrix}
		
	- 10.2 \% \\ +11.04 \% \\ - 12.90\%
	\end{pmatrix}\].
	
We are assuming only long positions.
 Consequently, for security $m_2$, which is anticipated to trade at a price lower than the agreed-upon spot price, brokers are expecting to incur losses.

\subsection{Mechanism} \label{mechanism}
The auction takes place in two rounds $t=1,2$.  Each broker $i\in\mathbb{N} $ submits consecutively a single\footnote{For simplicity we restrict our analysis to this class of auctions.} bid (fee) in basis points $\phi_i^t\in[0,1]$, with $\phi_i^1(s_i,z)$ to be the first-round bid and $\phi_i^2(s_i,\mathit{z^\prime})$ be the second-round bid for $\mathit{z^\prime}$ an update in public information. The fee is calculated ad valorem on the portfolio's value.
The standard tie-breaking rule (in which the winner is selected at
random) applies to both rounds.

\textit{First round.} The asset manager initiates simultaneously $\mathit{q}+1$ sealed auctions for the $\mathit{q}$ packages and the aggregate portfolio $\Theta$. Each local $\ell\in \mathbb{L}$ competes for the package $\theta_j$ that he is interested in and receives no extra utility from owning more than one package. Accordingly, each global $g\in \mathbb{G}$ competes for the aggregate $\Theta$ and receives no utility from owning a single package. The qualified brokers for the next round are $\mathit{q}$ \textit{local} winners and one \textit{global} winner with the lowest bids (first price). The first-round equilibrium always exists under the standard assumptions.

\textit{Second round.} At the outset, the asset manager updates the available public information to $\mathit{z^\prime}$, by revealing the winning bids of the previous round. Then, the qualified \enquote{local} winners of the first round, defined as $\mathbb{Q} \subset \mathbb{L}$  with cardinality $|\mathbb{Q}|=\mathit{q}$ i.e., the number of local packages, jointly compete against the qualified \enquote{global} winner $\mathit{g}$.

In all cases, the second-round bids  are bounded from above by the first-round bidding ($\phi_i^2 \leq \phi_i^1$, $\forall i\in \mathbb{N}$). In fact, this attribute rules out manipulability in the first round. The second round follows the rules of core-selecting auctions, with two possible outcomes\footnote{Without loss of generality ties are resolved.}: the global broker $\mathit{g}$ wins all packages as $\Theta$ when $\phi_g^2 <\sum\limits_{i\in \mathbb{Q}}\omega_i\phi_i^2$, and each local broker $i$ wins one package $\theta_j$ if $\phi_g^2 >\sum\limits_{i\in \mathbb{Q}}\omega_i\phi_i^2$.

The payoff of a local broker $i$, who wins a package $\theta_j$ with $m$ securities for the charged commission $c_i\in \mathbb{R}_+$ is given by: 

\begin{equation}\label{equation1}
\pi_i(\phi^2|s_i,z)=  	\mathbb{E} \big[ (\mathbf{\theta}_j \cdot p^*)\cdot c_i - \mathbf{\theta}_j \cdot \mathbb{E} [\Delta p] \big]\cdot \bm{1}_{\{\sum\limits_{i\in \mathbb{Q}}\omega_i\phi_i^2 < \phi_g^2\}}
\end{equation} 
\\
where the last term is an indicator function for $\sum\limits_{i\in \mathbb{Q}}\omega_{i}\phi_i^2 < \phi_g^2$, when the whole portfolio $\Theta$ is assigned to local brokers for execution. 

For each package $\mathit{\theta_j}$ the expected payoff of each local  broker $i$  results from the charged commission upon the trading value $( \mathit{\theta_j}  \cdot  p^* \cdot c_i$)  minus the potential losses from the price variation ($\mathit{\theta_j} \cdot \mathbb{E}[\Delta p(s_i, z)]$), if exists. We denote the private valuation of each local broker $i$ for $\mathit{\theta_j}$  with $\alpha_i=\mathbb{E}\Big[\dfrac{\mathit{\theta_j} \cdot \Delta p}{\mathit{\theta_j}  \cdot  p^*}|s_i,z\Big]$, with $\alpha_i \in \mathbb{R}$.

On the other hand, if the global broker wins the aggregate portfolio $\Theta$, for $\sum\limits_{i\in \mathbb{Q}}\omega_i\phi_i^2>\phi_g^2$, he charges $C=\sum\limits_{i\in \mathbb{Q}}\omega_{i}\phi_i^2$  and follows a similar payoff function. We denote the private valuation for $\Theta$ of each global broker $g$ with $\upsilon=\mathbb{E}\Big[\dfrac{\mathit{\Theta} \cdot \Delta p}{\mathit{\Theta}  \cdot  p^*}|s_g,z\Big]$, with $\upsilon \in \mathbb{R}$.  The bidder-optimal-frontier, for any local broker $i$, is satisfied when $\sum\limits_{i\in \mathbb{Q}}\omega_ic_i=\phi_g^2$.

The VCG pricing function $\mathit{c}(\phi_\ell^2,\phi_g^2)$, where  $\mathbf{\phi}_\ell^2=(\phi_1^2, \dots, \phi_q^2)$, is given by:

\begin{equation}\label{eq:1}
	\mathit{c}(\phi_\ell^2,
	\phi_g^2) =
	\begin{cases}
		(c_1^V, \dots,
		c_q^{V}, 0)
		& \mbox{if $\phi_g^2 >\sum\limits_{i\in \mathbb{Q}}\omega_i\phi_i^2$ ,}\\ (0, \dots, 0, C) &\mbox{if $\phi_g^2<\sum\limits_{i\in \mathbb{Q}}\omega_i\phi_i^2$ .}
	\end{cases}
\end{equation}
\\

where 
$c_i^V=\max\bigg\{0,\dfrac{\phi_g^2 -\sum\limits_{j \neq i}\omega_{j}\phi_{j}^2}{\omega_i}\bigg\}$

Respectively, the core-selecting pricing rule is given by:

\begin{equation}\label{eq:2}
	\mathit{c}(\phi_\ell^2,
	\phi_g^2) = \begin{cases}
		(c_1, \dots,
		c_q, 0)
		& \mbox{if $\phi_g^2 >\sum\limits_{i\in \mathbb{Q}}\omega_i\phi_i^2$  ,}\\
		(0,\dots,0, C) &\mbox{if $\phi_g^2<\sum\limits_{i\in \mathbb{Q}}\omega_i\phi_i^2$ .}
	\end{cases} 
\end{equation}
\\
such that $c_i\in[\phi_i^2,c_i^V]$   with $\sum\limits_{i\in \mathbb{Q}} c_i\leq\phi_g^2$ and $C\in\big[\phi_g^2, \sum\limits_{i\in \mathbb{Q}}\omega_i\phi_i^2 \big] $.

If the VCG outcome is in the core, no broker has the incentive to deviate from his truthful preferences and it is the only selected Pareto-dominant outcome \citep{Ausubel2002}. 
However, when the VCG is outside the core, a different pricing rule is necessary if we are to minimize the incentives for deviation \citep{Day2007a}.  

In the following, we present the two core-selecting pricing rules: the Nearest-VCG rule \citep{Day2012} and the Dynamic-Nearest-VCG, a slight modification of the former that we introduce to accommodate our two-round setup.  In both cases, the global broker receives a fee equal to $\sum\limits_{i\in \mathbb{Q}}\omega_{i}\phi_i^2$ upon winning. Whereas if local brokers win, they apportion $\phi_g^2$ as follows:

\begin{enumerate}[label*=\arabic*.]
	\item \textbf{Nearest-VCG rule} 
	
	This pricing approach was first introduced by \cite{Day2007a} and \cite{Day2012}. The fundamental notion is to select a point in the bidder-optimal-frontier that will minimize the Euclidean distance from the VCG outcome \citep{Ausubel2019}. For a finite set of locals, the payments for the weighted packages are divided into:
	
	\begin{equation} \label{eq:3}
		c_i(\phi_\ell^2 , 
		\phi_g^2)=(c_1^{V} - \Delta_1, \dots,c_{q}^V-\Delta_q,0) ,
	\end{equation}
	
	where $\Delta_i=\sum\limits_{i\in \mathbb{Q}}\omega_ic_i^V - \phi_g^2 $ is the minimum downward correction on the VCG outcome that corresponds to each local broker $i$.

	\item \textbf{Dynamic-NVCG (D-NVCG) rule}
	
	This rule selects a vector of fees in the \textit{bidder-optimal-frontier}  determined by local brokers' first-round bidding. The rationale is that overbidding incentives in the first round are penalized for deviating from the VCG pricing. 
	
	Suppose $\mathbb{Q}=\mathbb{Q}^u \cup \mathbb{Q}^d$ and  $\mathbb{Q}^u \cap \mathbb{Q}^d=\varnothing$, where $\mathbb{Q}^u=\{j\in \mathbb{Q}|\phi_j^1>c_j^V\}$ and $\mathbb{Q}^d=\{i \in \mathbb{Q}|\phi_i^1\leq c_i^V\}$. Then, for any bidder $i$, the final fees of all bidders are readjusted downwardly by $\epsilon_i= \phi_i^1-c_i^V$.

	\begin{equation}\label{eq:5}
		c_i(\phi_\ell^2,\phi_g^2)=\begin{cases}
			
			\left[c_i^V - \Delta_{i} \right]+ \dfrac{
				\sum\limits_{j\in \mathbb{Q}^u}\epsilon_j}{\sum\limits_{i\in \mathbb{Q}^d }\omega_i}         & \mbox{if $\phi_{i}^1\leq c_i^V$} \\
			\left[c_i^V - \Delta_i\right] -\epsilon_i & \mbox{if $\phi_i^1>c_i^V$}
		\end{cases}
	\end{equation}

	where $\Delta_i=\sum\limits_{i\in \mathbb{Q}}\omega_ic_i^V - \phi_g^2$ and $\epsilon_j=\omega_j(\phi_j^1-c_j^V)$.

\end{enumerate}

Each bidder with $\phi_i^1>c_i^V$ will receive a downward adjustment on the Nearest-VCG pricing equal to the deviation $\epsilon_i$, while a bidder with $\phi_i^1\leq c_i^V$ will be rewarded for his strategy in the first round with an increase in the Nearest-VCG fee.

\subsection{Examples}
We provide two examples of the implementation of the payment rules. 

\paragraph{Example 1} Assume an asset manager demands liquidity for a portfolio $\Theta$. He divides the portfolio into two packages: the first package is $\theta_1$ with a weighted-value $\omega_1=0.6$ over the nominal value of $\Theta$, and the second package is $\theta_2$ with a weighted-value $\omega_2=0.4$ over the nominal value of $\Theta$. 
The qualified winners of the first round are the local broker $1$ for the package $\theta_1$, the local broker $2$ for the package $\theta_2$, and the global broker $g$ for portfolio $\Theta$.

In the second round, let's suppose that brokers submit the following bids in basis points\footnote{The payment also is in basis points, recall that $\omega$ is a real number that takes values [0,1]}
\begin{center}
	\begin{tabular}{c c c c c }
		\multicolumn{2}{c}{Local 1} & \multicolumn{2}{c}{Local 2}   & Global\\
		$\phi_{\ell_1}^1$	&$\phi_{\ell_1}^2$ &$\phi_{\ell_2}^1 $ &$\phi_{\ell_2}^2$ & $\phi_g^2$\\
		27	&25 &19 &10&22\\
	\end{tabular}
\end{center}
\bigskip
Since the aggregate bid of the local brokers equals $\omega_1\phi_{\ell_1}^2+\omega_2\phi_{\ell_2}^2=19$, they win. Thus, the asset manager assigns to the local brokers to execute the portfolio $\Theta$ jointly.  

At this point, the question is how much the asset manager will pay each local broker to execute each assigned package. If the asset manager applies the VCG pricing rule, as defined in equation (\ref{eq:1}), the resulting payments to local broker $1$ and broker $2$ are $c_1^V=30$ and $c_2^V=17.5$ respectively. 
However, the asset manager's total payment $\omega_1 c_1^V+ \omega_2 c_2^V = 25$ would be higher than if the portfolio was assigned to the global broker. Thus, the weighted total payment of the two locals must not exceed the biding $\phi_g^2$ of the global bidder (\enquote{second-price} rule).

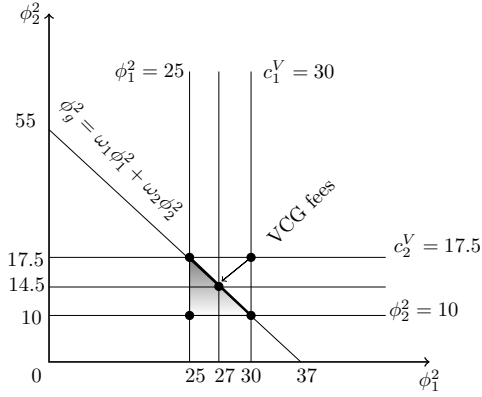
\begin{figure}[ht] 
	\centering
	\resizebox{6.5cm}{!}{
		
		\begin{tikzpicture}[yscale=1, xscale=1]
			\draw [<->,thick] (0,6) node (yaxis) [left] {$\phi_2^2$}
			|- (6.5,0) node (xaxis) [below] {$\phi_1^2$};
			\node [below left,black] at (0,0) {$0$};
			\node [below,black] at (-0.30,1) {$10$};

			\draw[->] (2.95 ,1.39)--(3.45,1.8)--(2.95 ,1.39);
			\shadedraw [shading=axis]  (2.4,1.8) --(2.4, 0.8) --(3.45,0.8) ;
 \node [below,black] at (-0.38,1.58) {$14.5$};  
			\node [below,black] at (-0.40,2) {$17.5$};
			\node [below,black] at (-0.40,4.4) {$55$};
			\node [below,black] at (2.5,0) {$25$};
             \node [below,black] at (3,0) {$27$};  
			\node [below,black] at (3.45,0) {$30$};
			\node [below,black] at (4.4,0) {$37$};
			\draw [black] (4.3,0)--(0,4); 
			\draw [black] (0,0.8)--(5.75,0.8);
   \draw [black] (0,1.29)--(5.75,1.29);
			\draw [black] (0,1.8)--(5.75,1.8);
			\draw[black]  (2.4,0)--(2.4,5);
   \draw[black]  (2.9,0)--(2.9,5);
			\draw[black]  (3.45,0)--(3.45,5);	
			\node[left,black] at (5,5) {$c_1^V=30$};
			\node[left,black] at (7.5,2) {$c_2^V=17.5$};
			\node [below,black] at (6.4,1.2) {$\phi_2^2=10$};
			\node [right,black] at (1,5) {$\phi_1^2=25$};
			\node [rotate=315] at (1.2,3.5) {$\phi_g^2=\omega_1\phi_1^2+\omega_2\phi_2^2$};	
			\draw [very thick, black] (3.45,0.8)--(2.4,1.8); 
			\filldraw[black] (2.4,1.8) circle (2pt) node[anchor=west]{};
			\filldraw[black] (3.45,0.8) circle (2pt) node[anchor=west]{};
			\filldraw[black] (2.4,0.8) circle (2pt) node[anchor=west]{};
   \filldraw[black] (2.9,1.3) circle (2pt) node[anchor=west]{};
			\filldraw[black] (3.45,1.8) circle (2pt);
			\node[rotate=45]at (4.3,2.6) {VCG fees};
	

	\end{tikzpicture}}
	\centering 
	\caption{Core point closest to VCG payments}
	
	\label{fig:1}

\end{figure}
\FloatBarrier

Figure \ref{fig:1} maps the broker's fees for which the coalition of locals is not blocked.The core itself can be graphed as a fee's space. The objective is to optimize the core outcome, considering the constraints for the values of $c_1$, $c_2$ from equation \eqref{eq:2}, and the bid of the global broker $\phi^2_g$. The values of $c_1$, $c_2$, must fall within a respective range, with $25 \leq c_1\leq 30$, $10 \leq c_2\leq 17,5$, subject to  $\omega_1c_1+\omega_2c_2\leq 22$. This means broker 1's bid must be fixed within the range [25,30]. The same applies to broker 2. 


One can readily notice that the constraints defining the core upwardly are the tie-breaking bids between the coalition and the global broker. The VCG rule defines the upper constraint for each local broker. The lower bounds on the local brokers' pricing values are their bids, consistent with the assumption of individual rationality \citep{Day2008b}.
Suppose that the local broker $1$ bids $\phi_1^2>30$. This outcome would be blocked by the global bid $\phi_g^2=22$ and broker 2's bid $\phi_2^2=10$.
The same applies if the local broker $2$ bids $\phi_2^2>17.5$. This outcome would be blocked by the global and local broker $1$.

Using the Nearest-VCG pricing rule from equation (\ref{eq:3}), we will minimize the distance from the VCG pricing rule to obtain an outcome that will be included in the core intervals, with downward adjustment equal to  $\Delta= \sum\limits_{i=1}\omega_ic_i^v - \phi_g^2=3$. In the core interval, local brokers' payoff is maximized at \textit{bidder-optimal-frontier} where the tie-break occurs.
Thus, the asset manager will pay the broker $1$ with the NVCG rule (equation \eqref{eq:3}) a commission $c_1=27$ bps and broker $2$ with a commission fee $c_2=14,5$ bps. These two fees ensure an outcome in the core and, at the same time, lie in the bidder-optimal-frontier, with $\omega_1c_1+\omega_2c_2= 22$. The point (27,14.5) which corresponds to the NVCG fees is depicted on the bidder-optimal-frontier in Figure \ref{fig:1}. The slope of the curve is defined by the weights.

One shortcoming of the Nearest-VCG pricing rule is that it has been designed for single-round auctions without encompassing the bidding behavior of the previous round. With the Dynamic-Nearest-VCG, incentives for bidding close to truthful valuations in the first round are rewarded. At the same time, those who misreport are \enquote{punished} by receiving a lower commission fee when the auction ends.  

For instance, the local broker 2 submits $\phi_2^1=19$ in the first round. This bid qualifies broker 1 for the second round, where the broker can reduce further or repeat it. Here, the bid of broker 2 is reduced in the second round by 9 bps while broker 1 reduces the bid by 2 bps. By the information released in the interim, each broker is updated for the price estimates of others.

According to equation (\ref{eq:5}), with the D-NVCG the asset manager will pay the Nearest-VCG prices minus any deviation that the VCG fee has from the  first-round bidding that is for broker $2$  a fee equal to $c_2=13$ bps and for broker $1$ $c_1=28$ bps. The intuition behind constructing the D-NVCG rule is to mitigate the free-riding incentives. Hence, the bidder $2$ will receive a reduced commission fee for the portfolio's execution since in the first round submitted a relatively high bid. This pricing rule might not have a direct effect on the asset manager's revenue since both rules (N-VCG and D-NVCG) lie in the bidder optimal frontier. However, it affects bidders' payoff.

In the next example, we illustrate the pricing rules for $\ell>2$ local brokers; the bids are quoted in bps:

\paragraph{Example 2} 

Assume that the qualified winners for the second round are five local brokers who compete against one global. Table \ref{tab1} presents each local broker $i$'s bid for each package $\theta_j$ with a weight $\omega_{j}$, respectively.  Since $\sum\limits_{i\in \mathbb{Q}}\omega_{i}\phi_{i}^2<\phi_g^2$, with $\sum\limits_{i\in \mathbb{Q}}\omega_{i}\phi_{i}^2=22.5$ and $\phi_g^2=25$, the asset manager assigns the portfolio's execution to local brokers. The VCG outcome $ c_i^V$ is calculated for each local broker $i$ based on equation (\ref{eq:1}) and presented in the relevant column.

\begin{table} 
	
	\centering  
	\resizebox{11cm}{!}{
		\begin{tabular}{c c c c c c c c c c} 
			\hline
	Local & Weights &\multicolumn{2}{c}{Bids} & VCG &  Core & Nearest & Dynamic \\
			
			\rule[-1ex]{0pt}{2.5ex} Brokers & $\omega_{i}$ & $\phi_i^1$ & $\phi_i^2$ & $c_i^V$  & Interval &  VCG  & NVCG\\
			\hline \hline
			\rule[-1ex]{0pt}{2.5ex} $1$ & $0.18$ & 25 & 20 & 33.88 & [20, 33.88]& 23.88  & 24.08 \\
			\hline
			\rule[-1ex]{0pt}{2.5ex} $2$ & $0.22$ & 30 & 21 & 32.36 & [21, 32.36] & 22.36 & 25.55 \\
			\hline
			\rule[-1ex]{0pt}{2.5ex}$3$ & $0.18$ & 36 & 22 &  35.89  & [22, 35.89] & 25.88 & 25.77 \\
			\hline
			\rule[-1ex]{0pt}{2.5ex} $4$ & $0.2$& 36 & 23 &  35.5 & [23, 35.5]& 25.5 & 25 \\
			\hline
			\rule[-1ex]{0pt}{2.5ex} $5$ & $0.22$& 30 & 26 & 37.36 &[26, 37.36]& 27.36  & 27.55 \\
			\hline \hline
			\rule[-1ex]{0pt}{2.5ex} $\mathit{Global}$ & & 28 & 25 &  & [22, 25]  &  &  \\
			\hline
	\end{tabular}} \caption{ 5 local and 1 global brokers}\label{tab1}
	
\end{table} 

\FloatBarrier

Similarly to the previous example, for the local broker $i$, a fee higher than $\phi_i^2>c_i^V$ is blocked by the coalition of locals given that others submit a bid equal to $\phi_{-i}^2$ and the global's bid $\phi_g^2=25$. The bidder-optimal-frontier is satisfied for $\sum\limits_{i=1}^{5}\omega_i\phi_{i}^2=25$ maximizing the local brokers' pay-offs, for every core interval defined by equation (\ref{eq:2}).

With the Nearest-VCG pricing rule from equation (\ref{eq:3}), the asset managers pay a commission fee to the local broker $1$ equal to  23.88, i.e., the local broker $i$ receives 10  bps of the commission fee less compared to the VCG rule. For the local brokers 2,3,4,5, the Nearest-VCG costs are presented in Table \ref{tab1}.

In this example, two local brokers have submitted a higher fee in the first round: local broker $3$ with $\phi_3^1>c_3^V$ and local broker $4$ with $\phi_4^1>c_4^V$, respectively. The suggested pricing rule restricts brokers from manipulating the outcome of the auction.

\begin{figure}
\centering
\begin{tikzpicture}
	\begin{axis}[xlabel={$\phi_i^2$}, ylabel={Pricing Rule ($c_i$)}, legend style={
			at={(0.5,-0.27)},
			anchor=north,
			legend columns=3
		}]
		\addplot[
		scatter,only marks,scatter src=explicit symbolic,
		scatter/classes={
			a={mark=square*,blue},
			b={mark=triangle*,red},
			c={mark=o,draw=black,fill=black}
		}
		]
		table[x=x,y=y,meta=class]{
		x y class label
		20 33.88 a
            21 32.36 a
		22 35.89 a
		23 35.5  a
	    26 37.36 a 
		20 23.88 b
		21 22.36 b
		22 25.88 b
		23 25.5 b
		26 27.36 b
		20 24.08 c
		21 22.55 c
		22 25.77 c
		23 25 c
		26 27.55 c
		};
		\legend{VCG,N-VCG,D-NVCG}
    \node[anchor=center] at (axis cs:2.9,4.5) {$\ell_1$};
        \node[anchor=center] at (axis cs:4.4,6.3) { $\ell_2$};
        \node[anchor=center] at (axis cs:1.8,3.8) {$\ell_3$};
         \node[anchor=center] at (axis cs:5,6.9) {$\ell_4$};
         \node[anchor=center] at (axis cs:8.8,10.7) {$\ell_3$};

	\end{axis}
\end{tikzpicture}

\caption{Pricing Rules for Package Bidding }  \label{fig:3}
\end{figure}
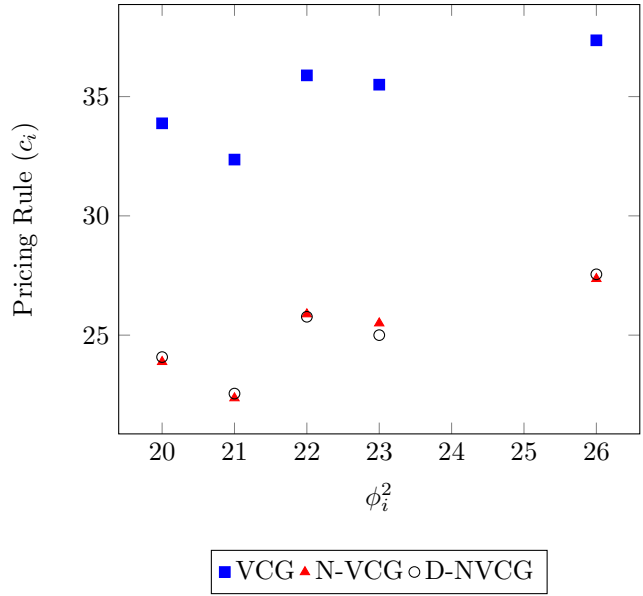

Applying the Dynamic-Nearest-VCG pricing rule from equation (\ref{eq:5}) for the local brokers $3$ and $4$ who have bidden excessively in the first round, the asset manager will pay the local broker $3$ a commission equal to 25.77 ($\epsilon_3=0.11$) and the local broker $4$ a commission equal to 25 ($\epsilon_4=0.5$) . Both local brokers $3$ and $4$ will bear an extra cost for their misreporting incentives in the first round.

Those local brokers who bid prudently in the first round - brokers 1,2 and 5 -  will receive an increased commission fee, with the incremental increase set at 0.19 bps ($\sum\limits_{j\in \mathbb{Q}^u}\epsilon_j=0.12$ and $\sum\limits_{i\in \mathbb{Q}^d }\omega_i=0.62$)) per broker. Thus, the local broker $1$ will receive a fee equal to 24.08 The same rule applies for local brokers $2$ and $5$. 
In Figure \ref{fig:3}, we can observe that the Dynamic-Nearest VCG pricing rule can yield higher fees to those locals who bid close to their truthful valuations in the first round, while simultaneously ensuring that the cost of execution remains low for the asset manager.
\section{Analysis} \label{Analysis}

We start our analysis by characterizing our mechanism in the second round as pivotal.  This means that the fee received by any local broker $i$ is equal to the loss imposed on other locals by adjusting $\phi_\ell^2$ to attribute $i$'s values. 
We define a bid $\phi_i^2$ submitted by broker $i$ as \textit{pivotal}, if and only if $\phi_g^2=\sum\limits_{i\in Q} \omega_{i}\phi_{i}^2$ holds and if for any $\gamma>0$, a bid $\phi_i^2-\gamma$ attributes a non-empty package $\theta_j$, while $\phi_i^2+\gamma$ yields the null package.

 A broker is pivotal if his report changes the auction outcome, in comparison to excluding the broker or attributing the null report to him, then the auction satisfies the pivotal pricing property \citep{milgrom2004,Ausubel2019}. In a local-global setting, this property is satisfied for any core-selecting auction.

\begin{lemma}[\cite{Ausubel2019}]\label{lemma1}
	Every core selecting auction satisfies the pivotal pricing property. 
\end{lemma}

It is not hard to prove that the Dynamic-NVCG pricing rule results in allocations belonging to the bidder-optimal-frontier and minimizes misreporting incentives. 

\begin{lemma} \label{lemma4.2}
	
The Dynamic-NVCG pricing rule, $c (\phi_\ell^2,\phi_g^2)$, lies on the bidder-optimal frontier.
\end{lemma}

\begin{proof} From equation (\ref{eq:5}), the total sum of local brokers' commission is:
\begin{align*}
	&	\sum_{i\in \mathbb{Q}^d}\omega_{i}c_{i}^V+ \sum_{i\in \mathbb{Q}^d}\omega_{i}\dfrac{
		\sum\limits_{j\in \mathbb{Q}^u}\omega_j(\phi_j^1 -c_j^V)}{\sum\limits_{i\in \mathbb{Q}^d }\omega_i} - \sum_{i\in \mathbb{Q}^d}\omega_{i} \sum_{i\in \mathbb{Q}^d}\Delta_{i} \\ &+ \sum_{j\in \mathbb{Q}^u}\omega_{j}c_j^V - \sum_{j\in \mathbb{Q}^u}	\omega_{j}(\phi_j^1 -c_j^V) - \sum_{j\in \mathbb{Q}^u}\omega_{j}\sum_{j\in \mathbb{Q}^u}\Delta_{j} \\
	=&	\sum_{i\in \mathbb{Q}}\omega_{i}c_{i}^V - \sum_{i\in \mathbb{Q}}\Delta_{i}=\phi_g^2.
\end{align*}
Consequently, the Dynamic-NVCG rule always lies on the bidder-optimal frontier. 
\end{proof}
We do not disregard the Nearest-VCG; instead, we are using it as a touchstone to improve incentive compatibility further and reduce the degree of manipulation freedom in a two-stage framework \citep{Francisco2001}. The following Proposition explains why the Dynamic-NVCG rule is optimal for distributing the commission of local brokers when some brokers have perverse incentives in the first round.

\begin{proposition} \label{proposition1}
	For any local bidder, $i$ bidding above $\mathbb{E}[c_i^V|s_i,z]$ in the first round is always a weakly dominated strategy.
\end{proposition}

\begin{proof} Suppose not. Then for any broker $i$, a  bidding strategy $\phi_i^{1\prime}\leq\mathbb{E}[c_i^V|s_i,z]$ is weakly dominated by $\phi_i^1>\mathbb{E}[c_i^V|s_i,z]$. By substitution in equation ($\ref{equation1}$) the pricing rule of ($\ref{eq:5}$) for $\pi_i(\phi_i^2|s_i,z^\prime)\geq\pi_i^\prime(\phi_i^2|s_i,z^\prime)$ we have:
	
 \begin{equation*}
   \theta_j\cdot p^*	\left(  c_{i}^V - \Delta_i-\epsilon_i\right) > \theta_j \cdot p^*	\left( c_{i}^V - \Delta_{i} + \dfrac{\sum\limits_{j\in Q^u }\epsilon_j}{\sum\limits_{i\in Q^d }\omega_i}	\right)	
 \end{equation*}

	Solving this inequality results in absurd. Thus, any bidding in the first round above the VCG price is a weakly dominated strategy. \end{proof}

Additionally, \cite{Voliotis2020} establish in Proposition 3 of their paper that, for all monotonically decreasing strategies $\phi$ in the first round, it is a symmetric equilibrium for each local broker to bid above the expected price differential. While the study focuses on a single portfolio in a two-stage design, the strategic behavior of brokers in the first round aligns with what is observed in our model.

Next, we provide some necessary assumptions.

\begin{assumption}
    The valuations $\alpha$ of local brokers are perfectly correlated.
\end{assumption}

\begin{assumption}\label{assumption1}
	For any winning bidder $i$ and any bidding vector  ($\phi^2_1,\dots,\phi^2_q,\phi_g^2$), the pricing function $c_i(\phi_1^2,\dots,\phi_q^2,\phi_g^2)$ is continuous in all bids and differentiable in bidder's $i$ bid.
\end{assumption}
\cite{Bosshard2017a} have proved that the N-VCG rule does not always satisfy the non-decreasing condition, yet in our model, both rules are increasing in the bidder's $i$. This means bidders are less prone to aggressive bidding. Lemma \ref{lemma_3}  confirms the revenue monotonicity by \cite{Milgrom2006} and \cite{Day2008b}. We prove that the revenue monotonicity holds in the bidder-optimal-frontier even for more than two items for sale \citep{Lamy2010}.

\begin{lemma} \label{lemma_3}
    Under the pivotal pricing property, the pricing function $c_i$ for NVCG and D-NVCG of the local bidder $i$ is increasing in the bidder's $i$ bid.  
\end{lemma}

\begin{proof}
    It follows immediately since ${c_i}^\prime>0$ in bidder's $i$ bid for both rules. This completes the proof of the lemma. 
\end{proof}



In the VCG mechanism, brokers bid their valuation truthfully, and it is their weakly dominant strategy with no surplus. The following lemma suggests
the global broker has a weakly dominant strategy in the second round. 
\begin{lemma} \label{lemma:3}
	
	Suppose that Assumption \ref{assumption1} is satisfied. Then, for the restricted second-round auction and for $\upsilon(s_g,z^\prime)>0$, $\phi_g^2=\min\{\phi_g^1,\upsilon(s_g,z^\prime)\}$
	is a weakly dominant strategy for the global broker. 
\end{lemma}

\begin{proof}
	By design, no broker can bid higher than his first-round bid in the second round. 
	For the VCG mechanism, it is a weakly dominant strategy for the global bidder to bid his \enquote{valuation}, which in our case equals to
	$ \upsilon (s_g,z^\prime)$. 
	The result follows directly. 
	\end{proof}

Whereas for the local bidders, bidding lower than their expected losses is always worse off, $\alpha_i(s_i,z^\prime)>0$.

\begin{lemma} \label{lemma:4}
	Suppose that Assumption \ref{assumption1}  and the pivotal pricing property are satisfied.
	Then, for each local broker $i$ any bid $\phi_i^2\in\big[0, \min\{\phi_i^1,\alpha_i(s_i,z^\prime)\}\big)$ is a weakly dominated strategy. 
\end{lemma}
\begin{proof}
	For an arbitrary broker $i$, let $\alpha_i(s_i,z^\prime)\geq\phi_i^1$. Then, for any strategy, $\phi_i^2\leq\phi_i^1$ is trivially weakly dominated and obtains negative surplus. Suppose now that for broker $i$, it is  $\alpha_i(s_i,z^\prime)<\phi_i^1$.
	For the local broker $i$ with $\hat{\phi_i^2}=\alpha_i(s_i,z^\prime)$ and 
	$\hat{\phi_i^2}>\phi_i^2$, we prove that bidding $\phi_i^2$ is weakly dominated. By \textit{Assumption} \ref{assumption1}, it will always result in  $c_i(\phi_\ell^2,\phi_g^2)\leq c_i(\hat{\phi_\ell^2},\phi_g^2)$, and from the pivotal pricing property, it follows:
	\small
	\begin{align*}
		\mathbb{E}\big[\theta_j \cdot p^* \cdot c_i(\phi_\ell^2,\phi_g^2)-\theta_j\cdot \Delta p(s,z^\prime)\big] &\leq\mathbb{E}\big[\theta_j\cdot p^* \cdot c_i(\hat{\phi_\ell^2},\phi_g^2) -\theta_j\cdot\Delta p (s,z^\prime)\big]\\
		&\leq\mathbb{E}\big[\theta_j\cdot p^* \cdot \frac{\theta_j \cdot\Delta p( s,z^\prime)}{\theta_j \cdot p^*} -\theta_j \cdot \Delta p(s,z^\prime)\big]=0.
	\end{align*}
	Thus, any $\hat{\phi_i^2}>\phi_i^2$
	is weakly dominated. 
	\end{proof}
 
Next, following \cite{Ausubel2019} we derive the necessary optimality conditions. Suppose that all local brokers $j\neq i$ bid according to the profile $(\phi_{j}^{2*})_{j \neq i}$. Let $H_i\equiv\mathit{H_i}\big[\phi_i^2,\alpha_i(s_i,z^\prime)\big]$ be the probability of winning for a local bidder $i$ who bids $\phi_i^2\in[\alpha_i(s_i,z^\prime),\phi_i^1\big]$, and its marginal probability  $h_i\equiv\mathit{h_i}\big [\phi_i^2,\alpha_i(s_i,z^\prime)\big]$:

\begin{center}
	\begin{align} \label{eqation4}
		H_i &=
		\Pr\big(\omega_i\phi_i^2 +\sum_{j \neq i}\omega_j\phi_j^{2*}\leq \upsilon(s_g,z^\prime) \big|\alpha_i(s_i,z^\prime) \big) \nonumber \\ \nonumber\\	h_i&= \dfrac{\partial H_i\big(\phi_i^2,\alpha_i(s_i,z^\prime)\big)}{\partial\phi_i^2}\
	\end{align}
\end{center}

Also, we denote the expected commission fee of each local bidder $i$ with $c_i\equiv C_i \big(\phi_i^2, \alpha_i(s_i,z^\prime)\big)$ and with $MC_i \equiv M C_i\big(\phi_i^2, \alpha_i(s_i,z^\prime)\big)$ the expected marginal commission when each local broker $i$ bids $\phi_i^2\in\big[\alpha_i(s_i,z^\prime), \phi_i^1\big]$.

\begin{align}\label{equation7}
	C_i &= \mathbb{E} \big[c_i \big(\phi_i^2, \sum_{j \neq i}\phi_j^{2*}, \upsilon(s_g,z^\prime)\big) \big|\alpha_i(s_i,z^\prime)\big]	\nonumber\\ \\
	M C_i&= \mathbb{E}\bigg[\dfrac{\partial c_i\big(\phi_i^2, \sum\limits_{j \neq i}\phi_j^{2*},  \upsilon(s_g,z^\prime)\big)}{\partial \phi_i^2} \big|\alpha_i(s_i,z^\prime) \nonumber\bigg].
\end{align}


The expected marginal commission expresses any change in the expected commission arising from
the incremental increase in the bidding $\phi_i^2$. For instance, if brokers anticipate a loss in the expected prices, they will counterbalance their payoff by moving their bid upwardly.  

Next, we define the first-order optimality conditions for the local broker's maximization problem on the steps of \cite{Ausubel2019}.

\begin{proposition}\label{Proposition2}
	Under Assumption \ref{assumption1} and the pivotal pricing property, the optimality condition for choosing $0<\phi_i^2\leq \phi_i^1$ for a local bidder $i$ is given by:
	\begin{align}\label{equation9}
		MC_i&= \bigg( \alpha_i(s_i,z^\prime) - \phi_i^2\bigg)h_i.
	\end{align}	 
	
\end{proposition}

\begin{proof}
	We apply the optimality condition on the expected payoff of equation (\ref{equation1}) upon the probability of winning.
	
	\[E[\pi_i(\phi_i^2|s_i,z^\prime)]=\bigg[\theta_j\cdot p^* \cdot c_i - \theta_j \cdot \mathbb{E}[\Delta p|s_i,z^\prime]\bigg]H_i,\]
	
	with $0\leq\phi_i^2\leq\phi_i^1$
	
	Following equations \eqref{eqation4} and \eqref{equation7} and maximizing with respect to $\phi_i^2$ yields the first-order condition: 
	\begin{align*}
		\dfrac{\partial E[\pi_i(\phi_i^2|s_i,z^\prime)]}{\partial \phi_i^2}&=
		\big[\theta_j \cdot p^* \cdot \dfrac{\partial c_i}{\partial\phi_i^2}\big] \cdot H_i + 	\bigg[\theta_j\cdot p^*\cdot c_i - \theta_j\cdot \mathbb{E}[\Delta p|s_i,z^\prime]\bigg]\cdot h_i \\
		&=
		\theta_j \cdot p^* MC_i + \theta_j\cdot p^* \cdot c_i \cdot h_i - \theta_j\cdot \mathbb{E}[\Delta p|s_i,z^\prime] \cdot h_i=0.
	\end{align*}

	By Lemma \ref{lemma:4}, $\phi_i^2$ is always nonnegative and the second-order conditions are trivially satisfied. Due to \textit{Assumption} \ref{assumption1} and the pivotal pricing property the following is in effect: 
	\[c_i\equiv c_i\big(\omega_i\phi_i^2,\sum\limits_{j\neq i}\omega_j\phi_j^{2*},\phi_i^2+\sum\limits_{j\neq i}\omega_j\phi_j^{2*}\big)=\phi_i^2.\] 
	
	Thus, it is easy to conclude that:
	\begin{align*}
		\theta_j \cdot p^* MC_i&= \theta_j\cdot \mathbb{E}[\Delta p|s_i,z^\prime]\cdot h_i - \theta_j\cdot  p^* \cdot \phi_i^2 \cdot h_i \\
		MC_i&= \bigg( \dfrac{\theta_j\cdot \mathbb{E}[\Delta p|s_i,z^\prime]}{	\theta_j\cdot p^*} - \phi_i^2\bigg)\cdot h_i.
	\end{align*} 
	
	\end{proof}
Intuitively,
if $MC_i<0$ it means that  
$\alpha_i(s_i,z^\prime) < \phi_i^2$, and 
broker $i$ is not included among the winners. Otherwise, if $\phi_i^2<\alpha_i(s_i,z^\prime) $, broker $i$ will have to increase his bidding fee to reach the optimal payoff where $MC_i=0$.

\begin{theorem}\label{theorem4.1}
	For each pricing rule exists an equilibrium where the bidding function of each broker is given by:
	\begin{enumerate} [label=(\alph*)]
		\item for the NVCG rule
		\begin{equation}\label{equation10}
			\phi_i^2=
			\begin{cases}
				\alpha_i(s_i,z^\prime) -
				\sigma_i \omega_i(q-1) &\mbox{,if $\alpha_i>0$}	\\
				0 &\mbox{,if $\alpha_i\leq0$}
			\end{cases}		
		\end{equation}
		\item for the D-NVCG rule
		\begin{equation}\label{equation11}
			\phi_i^2=\begin{cases}
				\alpha_i(s_i,z^\prime) -\sigma_i \omega_i(q-1) & \mbox{,if $\phi_i^1>c_i^V$ and $\alpha_i>0$}\\ 
				
				\alpha_i(s_i,z^\prime) - \sigma_i \omega_i\bigg[\dfrac{ \ell}{\sum\limits_{i\in Q^d}\omega_{i}}+ (q-1)\bigg] & \mbox{,if $\phi_i^1\leq c_i^V$ and $\alpha_i>0$} \\
				0& \mbox{,if $\alpha_i\leq0$}
			\end{cases}
		\end{equation}
		\\
		where $\dfrac{1}{\sigma_i}\equiv \dfrac{h_i}{H_i}$ is a reverse hazard rate and with $\ell$ to be the number of local bidders with $\phi_i^1>c_i^V$.  
	\end{enumerate}
\end{theorem}

\begin{proof}
	The optimality conditions are given by equation (\ref{equation9}).
	\begin{enumerate} [label=(\alph*)]
		\item From equation (\ref{eq:3}) of the \textit{Nearest-VCG rule}, the expected marginal commission of equation (\ref{equation7}) for a broker $i$ is:

		\begin{align*}
			MC_i&=[c^V-\Delta_i]^\prime H_i\\
			&= \omega_i (q-1) H_i.
		\end{align*}
		
		Replacing the above to (\ref{equation9}) we result to the equilibrium bid:
		
		\begin{align}\label{equation10}
			\centering
			\nonumber
			\phi_i^2 =\dfrac{\theta_j\cdot \mathbb{E}[\Delta p|s_i,z^\prime]}{	\theta_j\cdot p^*}- \omega_i(q-1)\dfrac{H_i}{h_i}.
		\end{align}
		\\
		\item For the D-NVCG
		if $\phi_i^1>c_i^V$ the equilibrium bidding is similar to the \textit{Nearest-VCG} of equation (\ref{equation10}). However, if there are $\ell$ number of locals with $\phi_j^1>c_j^V$ and bidder's $i$ bid in the first round is $\phi_i^1\leq c_i^V$, then the expected marginal commission from equation ($\ref{eq:5}$)  and for a symmetric $\omega_{i}=\omega_{j}$ is given by:

		\[M C_i=\omega_i\bigg[\dfrac{\ell}{\sum\limits_{i\in Q^d}\omega_{i}}+ (q-1)\bigg]H_i.\]

		By substitution in the optimality conditions of equations (\ref{equation9}) we conclude:
		
		\begin{equation*}
			\phi_i^2=\dfrac{\theta_j\cdot \mathbb{E}[\Delta p|s_i,z^\prime]}{	\theta_j\cdot p^*}- \omega_i\bigg[\dfrac{  \ell}{\sum\limits_{i\in Q^d}\omega_{i}}+ (q-1)\bigg]\dfrac{H_i}{h_i}.
		\end{equation*}
	\end{enumerate} 
	\end{proof}

In Theorem \ref{theorem4.1}, we proved the existence of equilibrium for Nearest-VCG and the Dynamic-Nearest-VCG. In both cases, we have shown that when the portfolio is sliced into many packages, the equilibrium bid of a winning broker is negatively affected by the number of packages and their size in the overall portfolio.

\textbf{Example.} Suppose that the distribution of global brokers is $F_g(\upsilon)=(\frac{\upsilon}{\Bar{\upsilon}})^\lambda$, where $\upsilon$ is the valuation of global broker and is drawn from the distribution on $[0,\Bar{\upsilon}]$ and $\lambda>1$, the parameter $\lambda$ controls the distributional strength of the global broker \citep{Ausubel2019}. We prove that it is an equilibrium for brokers to bid their valuation (i.e., $\phi^*=\alpha$).

Since the valuations of the local brokers are perfectly correlated in the second round, let denote $h_i(\phi_i^2, \alpha_i)=f_g (\omega\phi_i^2 + \sum_{j \neq i} \omega_j\phi_j^* (\alpha_i))$, where $\phi^*$ is the equilibrium bid, whereas $\phi^2_i$ is the highest bid (and simultaneously the lowest bid) for the symmetric equilibrium, with a slight abuse of notation, we denote $\phi^2_i$ by $\phi_i$.   \\
 Then the probability for the local brokers to win lies:

\[\omega_i\phi_i+\sum_{j \neq i} \omega_j{\phi^*_j} \leq\phi^2_g\leq \phi_i\]

\emph{Mutadis mutandis}, after the necessary manipulations, we derive the optimality conditions.

\begin{enumerate}[label=(\roman*)]

    \item NVCG rule:
    
\begin{align} \label{op1}
(\alpha_i -\phi_i) f_g(\omega_i\phi_i + \sum_{j \neq i} \omega_j{\phi^*_j}) = (q-1) \omega_i \big[F_g(\phi_i)- F_g(\omega_i\phi_i + \sum_{j \neq i} \omega_j{\phi^*_j})\big]
\end{align}

we rearrange \eqref{op1} and we denote 

\[ \frac{\lambda}{\omega_i(q-1)} (\alpha_i - \phi_i) = L(\phi_i).\]

where

\begin{align*}
   L(\phi_i)= [\omega_i\phi_i + \sum_{j \neq i} \omega_j{\phi^*_j}][\phi_i^\lambda(\omega_i\phi_i + \sum_{j \neq i} \omega_j{\phi^*_j})^{-\lambda} -1]
\end{align*}

If $\phi_i \in [0,\phi^*)$ then 

\[ \frac{\lambda}{\omega_i(q-1)} (\alpha_i - \phi_i) > L(\phi_i).\]

Similarly,
if $\phi_i \in [\phi^*, \phi_i^1]$ then 
\[ \frac{\lambda}{\omega_i   (q-1)} (\alpha_i - \phi_i) \leq  L(\phi_i).\]

For $\lambda>1$, $L'(\alpha_i)\geq 0$ and $L''(\alpha_i)\geq 0$, implying that L is an increasing function at the equilibrium point $\phi^*=\alpha_i$.


\item Dynamic-NVCG:

For $\phi^1_j > c_j^V $ the optimality condition is similar to \eqref{op1}. For $\phi^1_i \leq c_i^V $ the optimality condition changes:

\[(\alpha_i -\phi_i) f_g(\omega_i\phi_i + \sum_{j \neq i} \omega_j{\phi^*_j}) = \omega_i\bigg[\dfrac{\ell}{\sum\limits_{i\in Q^d}\omega_{i}}+ (q-1)\bigg] \big[F_g(\phi_i)- F_g(\omega_i\phi_i + \sum_{j \neq i} \omega_j{\phi^*_j})\big] \]

We rearrange

\begin{equation*} \label{op2}
\frac{\lambda \sum\limits_{i\in Q^d}\omega_{i} }{\omega_i [\ell + (q-1)\sum\limits_{i\in Q^d}\omega_{i} ]} (\alpha_i - \phi_i) =  L(\phi_i),
\end{equation*}

where 
\begin{align*}
  L(\phi_i)= [\omega_i\phi_i + \sum_{j \neq i} \omega_j{\phi^*_j}][\phi_i^\lambda(\omega_i\phi_i + \sum_{j \neq i} \omega_j{\phi^*_j})^{-\lambda} -1].
\end{align*}

Similarly to NVCG,  $L'(\alpha_i)\geq 0$ and $L''(\alpha_i)\geq 0$, implying that L is an increasing function at the equilibrium 
 point $\phi^*=\alpha_i$.

\end{enumerate}

\section{Conclusions} \label{concludes}

This research has studied a stylized model for a portfolio auction, which is divided into packages. We design a \enquote{local-global} environment with a finite set of locals, each interested in a single package. Our mechanism allows the asset manager to engage many brokers in the auction process, resulting in lower transaction costs. The information update in the interim for others' valuations mitigates the \enquote{winner's curse} \citep{Voliotis2020} and increases the broker's trust in the sense that the auction's rules have been followed. The design is a simple iterative mechanism with transparency in the auction process and eliminates any computational complexity for the winner's determination problem.

We introduce a dynamic setup for the Nearest-VCG pricing rule for a two-round auction. This new pricing setup aligns brokers' incentives to lower bids, mitigating the first round's free-riding opportunities. Using an endogenous reference rule for the expected VCG pricing outcome, the brokers are motivated to submit bids close to their truthful valuations in the first round, squeezing execution costs downward.

The Dynamic-NVG's primary goal of minimizing strategic misreporting and overbidding aligns with the challenges inherent in general package auctions. Transitioning to a generalized approach necessitates careful consideration of potential customizations to accommodate the diverse complexities, specifications, and constraints characterizing such auctions.

In conclusion, while our study has provided valuable insights within the local-global model, the prospect of generalizing Dynamic-NVCG for a wider spectrum of auctions remains an intriguing direction for future research. This extension has the potential to enhance the understanding of strategic behavior in diverse combinatorial auction scenarios.


	
	\bibliographystyle{apacite}
	\bibliography{Archix}
	
	\end{document}